\documentclass[11pt]{article}

\usepackage{amsmath}
\usepackage{enumitem}
\usepackage{amssymb}
\usepackage{amsfonts}

\usepackage[utf8]{inputenc}
\usepackage[margin = 1in]{geometry}

\usepackage{xifthen}
\usepackage{textcomp}
\usepackage{algorithm}
\usepackage{authblk}
\usepackage[noend]{algpseudocode}
\usepackage{amsthm}
\usepackage{xspace}
\usepackage{xargs}

\usepackage{faktor}

\let\E\undefined

\usepackage[classfont = bold]{complexity}
\usepackage{hyperref}
\usepackage{mathtools}
\usepackage[usenames]{color}
\usepackage[svgnames]{xcolor}
\usepackage{upgreek}
\usepackage{tikz}
\usepackage{bussproofs}
\usepackage{todonotes}
\usepackage{thmtools}
\usepackage{thm-restate}
\usepackage{tcolorbox}
\usepackage{etoolbox}
\usepackage{xargs}
\usepackage{xparse}
\usepackage[style = alphabetic, maxbibnames = 99, backend = bibtex]{biblatex}
\let\E\undefined

\usepackage{pgfplots}
\pgfplotsset{compat=1.16}
\usetikzlibrary{positioning}
\usetikzlibrary[patterns, patterns.meta]
\usetikzlibrary[fadings]
\usepackage{color,hyperref,cleveref,bm}

\hypersetup{
    colorlinks=true,
    linkcolor=blue,
    filecolor=magenta,      
    urlcolor=cyan,
    citecolor=cyan,
}

\theoremstyle{definition}
\newtheorem{theorem}{Theorem}[section]
\newtheorem{lemma}[theorem]{Lemma}
\newtheorem{proposition}[theorem]{Proposition}

\newtheorem{definition}[theorem]{Definition}
\newtheorem{corollary}[theorem]{Corollary}
\newtheorem{conjecture}[theorem]{Conjecture}
\usepackage{enumitem}

\DeclareMathOperator*{\E}{\mathbb{E}}
\newcommand{\Maj}{\textsc{Maj}}
\renewcommand{\epsilon}{\varepsilon}

\newcommand{\newauthor}[3]{
    \newcounter{#1comment}
    \setcounter{#1comment}{1}
    \expandafter\newcommand\csname #1comment\endcsname[1]{%
            \par\noindent
            \todo[inline, size = \small, backgroundcolor = {#3}, caption = {}]{
                \arabic{#1comment}:
                {##1} --~\textbf{#2}
            }
            \addtocounter{#1comment}{1}
    }
    \expandafter\newcommand\csname #1changed\endcsname[1]{%
        \ifdraft
            \colorbox{#3}{
                ##1
            }%
        \else
            ##1%
        \fi
    }
}

\newauthor{ds}{Dmitry}{red!30}
\newauthor{ar}{Artur}{blue!30}
\newauthor{yf}{Yuval}{green!30}
\newauthor{il}{Itai}{orange!30}

\title{Sampling and Certifying Symmetric Functions}

\newif\ifanonymous
\anonymousfalse

\ifanonymous
\else
\author[1]{Yuval Filmus}
\author[2]{Itai Leigh}
\author[3]{Artur Riazanov}
\author[4]{Dmitry Sokolov}
\affil[1]{Technion --- Israel Institute of Technology, yuvalfi@cs.technion.ac.il}
\affil[2]{Tel-Aviv University; University of Copenhagen; itai.leigh@mail.huji.ac.il}
\affil[3]{EPFL, tunyash@gmail.com}
\affil[4]{EPFL, sokolov.dmt@gmail.com}
\fi

\date{May 2022}

\begin{document}

\maketitle
\begin{abstract}
    A circuit $\mathcal{C}$ samples a distribution $\bm{X}$ with an error $\epsilon$ if the statistical
    distance between the output of $\mathcal{C}$ on the uniform input and $\bm{X}$ is $\epsilon$. We
    study the hardness of sampling a uniform distribution over the set of $n$-bit strings of Hamming
    weight $k$ denoted by $\bm{U}^n_k$ for \emph{decision forests}, i.e.\ every output bit is computed as
    a decision tree of the inputs. For every $k$ there is an $O(\log n)$-depth decision forest sampling
    $\bm{U}^n_k$ with an inverse-polynomial error \cite{Viola12,Czu15}. We show that for every $\epsilon
    > 0$ there exists $\tau$ such that for decision depth $\tau \log (n/k) / \log \log (n/k)$, the error
    for sampling $\bm{U}_k^n$ is at least $1-\epsilon$.  Our result is based on the recent robust
    sunflower lemma \cite{ALWZ21,Rao19}.

    Our second result is about matching a set of $n$-bit strings with the image of a $d$-\emph{local}
    circuit, i.e.\ such that each output bit depends on at most $d$ input bits. We study the set of all
    $n$-bit strings whose Hamming weight is at least $n/2$. We improve the previously known locality
    lower bound from $\Omega(\log^* n)$ \cite{BDKMSSTV13} to $\Omega(\sqrt{\log n})$, leaving only a
    quartic gap from the best upper bound of $O(\log^2 n)$. 
\end{abstract}

\section{Introduction}
Studying the hardness of sampling has been proposed in the paper \cite{Viola12}, which spurred an active
line of research ~\cite{Viola12Turing,LV11,BIL12,Viola14,Viola16,Viola20,GW20,Vio21,BDFIKS22,CGZ22}. The
basic setting is the following: we are given an infinite supply of independent uniform random bits as
input, and our goal is to design a circuit with multiple output bits whose output is close in statistical
distance to a given target distribution. Sampling circuit constructions have been applied in
cryptography~\cite{IN96,BIVW16} and algorithms~\cite{Hag91}. Sampling hardness results\footnote{That is,
    showing that any circuit from a certain class produces a distribution that is far from the target.}
have been instrumental in inspiring and improving two-source extractor
constructions~\cite{Viola14,CZ16,Coh16} (see \cite{Viola20} for an extended discussion), and have yielded
lower bounds on succinct data structures~\cite{Viola12,Viola20,Vio21}.

Sampling hardness results are more challenging than computational ones. For example, while it is known since
Smolensky's classical work~\cite{Smo87} that parity requires an exponential number of gates to compute by
an $\AC^0[3]$ circuit, no hard distributions are known for the circuit class $\AC^0[p]$ for any $p$, and
while $\AC^0$ requires exponentially many gates to compute parity~\cite{Has87}, a random vector
with parity $0$ can be sampled by an $\NC^0$ circuit.  
Moreover, this distribution can be sampled by a
$2$-local circuit, in the sense that each output bit depends only on two input bits \cite{Bab87,Kil88}.
A very simple mapping achieves this: 
\( (x_1, \dots, x_n) \mapsto (x_1, x_1 \oplus x_2, x_2 \oplus x_3, \dots, x_{n-1} \oplus x_{n}, x_n). \)
A more general and striking fact is that $\AC^0$ can sample random permutations and consequently all
distributions of form $(\bm{X}, f(\bm{X}))$ where $\bm{X}$ is uniform over $\{0,1\}^n$ and $f$ is
\emph{symmetric} i.e.\ its value depends only on the Hamming weight of the input~\cite{Viola12}. 

We conjecture that the power of $\NC^0$ in regard to sampling symmetric distributions is in essence
limited to the parity example above. Observe that a function is computable by an $\NC^0$ circuit
if and only if it is \emph{$O(1)$-local} i.e.\ each of its output bits depends on at most a constant number of input bits. For simplicity, we focus only on uniform distributions with symmetric support: let $\bm{U}^n_S$ be the uniform distribution over strings in $\{0,1\}^n$ with Hamming weight in the set $S \subseteq \{0, \dots, n\}$.
\begin{conjecture}
    \label{cnj:main}
     For every $d \in \mathbb{N} ,\epsilon \in (0,1)$ for all large enough $n$, if
    $\bm{X}$ is samplable by a $d$-local function and is $\epsilon$-close to $\bm{U}^n_S$ for some $S
    \subseteq \{0, \dots, n\}$, then $\bm{X}$ is $O(\epsilon)$-close to $\bm{U}^n_T$, where $T$
    is one of the following: $\{0\}, \{n\}, \{0,n\}, \{0,2,4,\ldots\}, \{1,3,5,\ldots\}, [n]$.
\end{conjecture}

Our results on sampling slices (namely, \Cref{thm:main}) imply this conjecture for all sets $S$ which
only contain small values, in the sense that $\max_{x \in S} x = o(n)$. 

\paragraph{Quantum separations.} A stronger version of \Cref{cnj:main} would identify the family of sets
$S$ such that every $\NC^0$-samplable distribution is $1-o(1)$-far from $\bm{U}_S^n$. There exists a set
$S$ for which this implies a separation between $\NC^0$ and $\QNC^0$ for sampling. This is due to the
recent partial separation in \cite{BPP23}: they show that there exists a symmetric function $f$ such that
$(\bm{X}, f(\bm{X}))$ for uniform $\bm{X} \sim \{0,1\}^n$ can be sampled by a $\QNC^0$ circuit. 
%
Observe, however, that if an $\NC^0$-samplable distribution $\bm{Y}$ is at distance $\eta$ from $(\bm{X},
f(\bm{X}))$, then the first $n$ bits of $\bm{Y}$ are $(1/2+\eta+o(1))$-close to the uniform distribution
over $f^{-1}(0)$, due to the fact that the function $f$ used in~\cite{BPP23} is almost balanced (in the
sense that $|f^{-1}(0)| = (1+o(1)) \cdot |f^{-1}(1)|$). Now, if the uniform distribution over $f^{-1}(0)$
is not $\NC^0$-samplable within the distance $1-\Omega(1)$, we get the separation. \Cref{cnj:main}
implies a weaker lower bound for the distance: a constant instead of a function approaching $1$, but most
of the known lower bounds for distribution sampling have very strong distance guarantees. 

\paragraph{Local certificates.} One interesting class of sets which is also not covered by our and prior
results is $M_a = \{x \in \{0, \dots, n\} \mid x \bmod a = 0\}$, where $a$ is a constant. However, the
simple construction for sampling parity-$0$ vectors can be adapted to \emph{match the support} of any
$\bm{U}_{M_a}^n$, i.e.\ generate a (not necessarily uniform) distribution whose support is exactly the
set of strings with Hamming weight in $M_a$. This construction was given in
\cite[Proposition~3.1]{BDKMSSTV13}, where it was presented as a \emph{proof system}. The idea is to
interpret the input bits as a \emph{certificate} that the output is in the target language (in this case,
all $n$-bits strings whose Hamming weight is divisible by $a$). This connection motivates our study of
locality in the context of proof systems. We drastically simplify and improve the locality lower bound
of~\cite{KLMS16} for the language of $n$-bit strings whose Hamming weight is at least $n/2$ (in other
words, $1$-inputs of the majority function).

\subsection{Sampling Slices}
 The $k$-\emph{slice} is the set of all $n$-bit strings of Hamming weight $k$. We denote the uniform distribution over the $k$-slice by $\bm{U}_k^n$. A simple computation (see \Cref{sec:moreover-of-thm-main}) shows that $\bm{U}_S^n$ and $\bm{U}_{\max S}^n$ are close in statistical distance whenever $\max S = o(n)$. This means that in order to show the hardness of sampling from symmetric distributions over sublinear-Hamming-weight strings it is sufficient to study $\bm{U}_k^n$ for $k=o(n)$. 

Although in the context of \Cref{cnj:main} it is sufficient to lower bound the locality of a sampler, in
the context of slices the more natural complexity measure is \emph{decision depth}. A function $f\colon
\{0,1\}^m \to \{0,1\}^n$ is computable by a \emph{decision forest} of depth $d$ if every output bit of
$f$ can be computed by a decision tree of depth $d$, i.e.\ with at most $d$ \emph{adaptive} queries to
the input (in contrast, a $d$-local function is computed by $d$ \emph{non-adaptive}
queries). Viola~\cite[Lemma~6.4]{Viola12} shows that a decision depth $d$ sampler for $\bm{U}_k^n$ can be
obtained from a depth-$d$ \emph{switching network}. Czumaj~\cite[Theorem~3.7]{Czu15} proves the existence
of such switching networks with $d=O(\log n)$. The following theorem shows that for $k=o(n)$, this
construction is almost tight.

\begin{restatable}{theorem}{MainThm}
    \label{thm:main}   

    Suppose that $\bm{U}_k^n$ can be sampled with decision depth $d$ and error $\eta$ in variation
    distance.
    \begin{enumerate}
        \item \label{item:sublinear} For every $\epsilon > 0$ there exists a constant $\tau$ such that $d \le \tau \log (n/k) / \log\log (n/k)$ implies $\eta \ge 1 - \epsilon$.
        \item \label{item:subpoly} There exists a constant $\tau$ for which the following holds. For every $\epsilon \in (0,1)$, if $k \in [\log_2 n,
            2^{\log^{1-\epsilon} n}]$ and $d \le \tau \log^{\epsilon} n$, then $\eta = 1-n^{-\Omega(k)}$. The same bound holds for $k \in [1, \log_2 n)$ and $d \le \tau \log^{\epsilon} n / \log \log n$.
    \end{enumerate}
    Moreover, \Cref{item:sublinear} holds for any $\bm{U}_S^n$ with $\max_{x \in S} x = k$ or
    $\min_{x \in S} x = n-k$.
\end{restatable}

The first key observation in our proof of \Cref{thm:main} is that it is sufficient to prove it for $k=1$,
namely:
\begin{restatable}{theorem}{MainThmUone}
    \label{thm:U1-main}
    There exists a constant $\tau > 0$ such that any distribution sampled with decision depth $\tau \log
    n/\log \log n$ is $(1-n^{-\Omega(1)})$-far from $\bm{U}^n_1$.
\end{restatable}

To see why this implies \Cref{item:sublinear}, observe that the marginal distribution of the first $n/k$
bits of $\bm{U}_k^n$ is $(1-1/e+o(1))$-close to $\bm{U}_1^{n/k}$. \Cref{thm:U1-main} then implies the
distance lower bound $1 / e - o(1)$ for sampling $\bm{U}_k^n$ with depth $\tau \log (n/k) / \log\log
(n/k)$. The $1-\epsilon$ distance lower bound for every $\epsilon$ is achieved by generalizing \Cref{thm:U1-main} so it applies to distributions of the form ``first $\Theta(n/k)$ bits of $\bm{U}_k^n$'' directly.

\Cref{item:subpoly} is implied by another reduction from \Cref{thm:U1-main}. Suppose we have a depth-$d$
sampler for $\bm{U}_k^n$. Using this sampler, we can construct a depth-$kd$ sampler for
$\bm{U}_1^{\binom{n}{k}}$ as follows. Identify each output bit with a unique subset of $[n]$ of size $k$,
and assign to it the conjunction of the corresponding $k$ bits in the output of the sampler for
$\bm{U}_k^n$. It is easy to see that the resulting distribution has the same distance to
$\bm{U}_1^{\binom{n}{k}}$ as the initial sampler has to $\bm{U}_k^n$. 

\subsubsection{Technique for Proving \texorpdfstring{\Cref{thm:U1-main}}{Theorem \ref{thm:U1-main}}}

The locality of a sampler is the maximum number of input bits that an output bit depends on. The locality
is always bounded from above by the decision depth. As a warm-up, let us discuss an $\Omega(\log \log n)$
locality lower bound for sampling $\bm{U}_1^n$ which is close in spirit to \cite[Theorem~3]{Viola20}.

The main idea in the locality lower bound is a \emph{hitting set} versus \emph{independent set}
dichotomy: for a $d$-local source, it is easy to see that either there are $\tau 2^d$ independent output
bits, or there is a hitting set of input bits of size $\tau d 2^d$ such that every output bit depends on
one of the bits in this set. In the former case, we can show that it is very likely that at least two of
the independent bits evaluate to $1$, since for each output bit its probability to be $1$ is at least
$2^{-d}$.\footnote{There is a caveat that some bits might be identically zero, but it is not a real
    issue, since there cannot be too many of them.} In the latter case, by fixing the hitting set bits in
every possible way, we observe that our source is a mixture (convex combination) of $2^{\tau d 2^d}$ many
$(d-1)$-local sources. If we show that $(d-1)$-local sources must be $(1 - \epsilon)$-far from the target
distribution, then our source is $(1 - \epsilon 2^{\tau d 2^d})$-far by
\cite[Corollary~18]{Viola20}. Picking $d = \delta \log \log n$ for small enough $\delta$ yields a
$1-o(1)$ lower bound on the distance to the target distribution.

\medskip

In order to improve this lower bound from $\log\log n$ to $\log n / \log\log n$, we introduce several new
ideas. 

\emph{Monotonization.} We observe that it is in some sense sufficient to deal with sources where each bit
is a \emph{monotone term} in the input bits. The intuitive reason is that the expected number of output
bits evaluating to $1$ is $2^{-d} \cdot n$\footnote{Again, this is not always true, since there are
    identically zero outputs.}, so there exists an assignment where that many bits evaluate to $1$. By
focusing on those bits and replacing them with terms corresponding to the satisfying assignments, we show
that it is likely that at least two of these terms evaluate to $1$.

\emph{Sunflowers.} Let us pretend that all of the output bits are monotone terms. For $i \in [n]$, let $N_i$ be the set of inputs mentioned by the term of the $i$th output bit. We find a sunflower
$\mathcal{S} \subseteq \{N_1, \dots, N_n\}$, i.e.\ there exists a kernel $K$ such that the intersection
of any pair of sets in $\mathcal{S}$ is $K$. If the kernel is fixed to $1$, the output bits in the
sunflower become independent, so if the sunflower is large enough, it is likely that at least two of them
evaluate to $1$. For some small enough $d = \Omega(\sqrt{\log n})$, a large sunflower always exists among
any $n$ sets. Moreover, we can cover all but $o(n)$ output bits with sunflowers. Then we have the
following dichotomy: either all kernels evaluate to $0$, so the source is likely to be identically zero,
or at least one kernel evaluates to $1$, which makes it very likely  that at least two output bits from
the corresponding sunflower evaluate to $1$. 

Using \emph{robust sunflowers} instead of classical sunflowers, we obtain a lower bound of
$\Omega(\log n / \log \log n)$ on the decision depth.


\subsection{Local Proof Systems}
Local proof systems, introduced in \cite{BDKMSSTV13} and further studied in \cite{KLMS16}, are defined as
follows: a local proof system for a language $L$ is an $\NC^0$-circuit family $C_n$ such that $L \cap
\{0, 1\}^n$ is exactly the set of all possible outputs of $C_n$. A language $L$ has a \emph{$d$-local}
proof system if for each $n$ there exists a $d$-local function whose image is $L \cap \{0, 1\}^n$. In
relation to the sampling, it can be viewed as follows: the ``sampler'' needs to match the support of the
given distribution exactly, but we do not care about matching the actual probabilities. 

The hardness landscape in this setting is different from the setting of sampling distributions. The most
notable difference is that the language of strings with Hamming weight divisible by $p$ always has an
$O(1)$-local proof system, while we conjecture that sampling from the uniform distribution over this
language requires a super-constant locality, unless $p = 1$ or $p = 2$.

Our main contribution is in improving the locality lower bound for another symmetric language:
$\Maj^{-1}(1) \coloneqq \{x \in \{0, 1\}^n \mid |x| \ge n / 2\}$. The previous best locality lower bound
for proof systems for this language was $\Omega(\log^* n)$ \cite{KLMS16}, with a very complex proof which
we expose in \Cref{sec:majority-log-star}. We simplify this proof and improve the locality lower bound to
$\Omega(\sqrt{\log n})$, which is only polynomially smaller than the current best upper bound of
$O(\log^2 n)$. The key idea is to consider proof systems with bounded locality of both input and output
bits. For such proof systems it is easy to derive very strong requirements on locality. These can then be
used to count the number of input-output bit pairs where the input affects the output, which yields the
output locality lower bounds.

\subsection{Switching Networks}
The technique behind \Cref{thm:main} breaks down for linear slices $\bm{U}_{\alpha n}^n$. Does it mean there is a low-depth sampler for such distributions? 
The only construction of a sampler we have is based on switching networks \cite[Theorem~3.7]{Czu15}.

A switching network of depth $d$ is a layered graph with $d$ layers and $n$ nodes in each layer. The edges of the switching network do not cross between layers, and in each layer, the edges form a partial matching. A switching network defines a process over permutations of $[n]$: we start with the identity permutation, and then, for each layer, we toss a coin for each edge in the matching of this layer, and if the coin comes up heads, we transpose the endpoints of the edge.

Currently, the best upper bound on the depth of a switching network which produces a distribution close to the uniform distribution over all permutations is $O(\log^2 n)$ \cite{Czu15}. The situation is better if the switching network only needs to shuffle sequences of zeroes and ones: the input is now $1^k 0^{n-k}$. \cite[Theorem~3.7]{Czu15} gives a $O(\log n)$ depth upper bound for this case (the construction is randomized), \cite{GT14} gives an \emph{explicit} lower bound for generating $\bm{U}_k^n$ for $k \le \sqrt{n}$. 

It is almost immediate that a switching network that samples $\bm{U}_k^n$ within a non-trivial distance must have depth $\Omega(\log (n/k))$: each input bit of a network of depth $d$ has at most $2^d$ potential positions that it can take in the output, so if $d=o(\log (n/k))$, the switching network produces a distribution where only $o(n)$ bits have a non-zero probability to have value $1$, which is $(1-o(1))$-far from $\bm{U}_k^n$.

In \Cref{sec:switching} we show that switching networks that produce a distribution close to $\bm{U}_{\alpha n}^n$ must have depth $\Omega(\log \log n)$. 
We use the following properties of samplers that are constructed from switching networks of depth $d$: the first is that each input bit of such a sampler affects at most $2^d$ output bits, the second is that the error is one-sided, i.e.\ such a sampler never outputs a string outside the domain of $\bm{U}_{\alpha n}^n$. The second property highlights the similarity with local certificates, and indeed our lower bound proof uses similar ideas.

\subsection{Further Research}

Our results on sampling slices with decision forests taken together with results of Viola \cite{Viola12} are summarized in \Cref{fig:results-summary}.
\begin{figure}[htbp]
    \centering
    \begin{tikzpicture}
        \begin{axis}[
            width = \textwidth,
            height = 20cm,
            domain = -4:0.7,
            ymin = -0.1,
            ymax = 0.5,
            samples = 100,
            axis lines = none,
            enlargelimits = false
            ]
            \def\cutone{-2}
            \def\cuttwo{-1}
            \def\cutzero{-2.9}
            \addplot [draw=black!50,fill=red!30,domain=\cuttwo:0.3] {1/sqrt(2*pi)*exp(-x^2/2)/5}
            \closedcycle;
            \addplot [draw=black!50,fill=orange!30,domain=-4:\cuttwo] {1/sqrt(2*pi)*exp(-x^2/2)/5}
            \closedcycle;
            \addplot [draw=black!50,fill=yellow!40,domain=-4:\cutone] {1/sqrt(2*pi)*exp(-x^2/2)/5}
            \closedcycle;
            \node[rounded corners,left=1mm,fill=white,fill opacity=0.7,text opacity = 1] at  (axis
            cs:\cutone,0.015) {$k\le 2^{\log^{1-\epsilon} n}$} ;
            \node[rounded corners,left=1mm,fill=white,fill opacity=0.7,text opacity = 1] at  (axis
            cs:\cuttwo,0.015) {$k=o(n)$};
            \node[left=1mm,fill=white,fill opacity=0.7, rounded corners,text opacity = 1] at (axis
            cs:0.3,0.015) {$k=\Theta(n)$};
            \def\rowone{-0.012}
            \def\rowtwo{-0.031}
            \def\rowthree{-0.051}
            \node[rounded corners,black,left=1mm,fill=white,opacity=0.9] at  (axis cs:\cutzero,\rowone)
            {\Cref{thm:main}} ;
            \node[rounded corners,black,left=1mm,fill=white,opacity=0.9] at  (axis cs:\cutone,\rowone)
            {\small $d = \Tilde{\Omega}(\log^{\epsilon} n)$} ;
            \node[rounded corners,black,left=1mm,fill=white,opacity=0.9] at  (axis cs:\cuttwo,\rowone)
            {\small $d = \Tilde{\Omega}(\log(n/k))$} ;
            \node[rounded corners,black,left=1mm,fill=white,opacity=0.9] at  (axis cs:\cutone,\rowtwo)
            {\small $\Delta = 1 - n^{-\Omega(k)}$} ;
            \node[rounded corners,black,left=1mm,fill=white,opacity=0.9] at  (axis cs:\cuttwo,\rowtwo)
            {\small $\Delta = 1 - \delta$} ;
            \draw[densely dotted] (axis cs:\cutone,0) -- (axis cs:\cutone,\rowthree-0.012);
            \draw[densely dotted] (axis cs:\cuttwo,0) -- (axis cs:\cuttwo,\rowthree-0.012);
            \draw[densely dotted] (axis cs:\cutzero,0) -- (axis cs:\cutzero,\rowthree-0.012);
            \draw[densely dotted] (axis cs:-4,\rowthree+0.012) -- (axis cs:0.3,\rowthree+0.012);
            \draw[densely dotted] (axis cs:-4,\rowthree-0.012) -- (axis cs:0.3,\rowthree-0.012);

            \node[rounded corners,black,left=1mm,fill=white,opacity=0.9] at  (axis cs:\cutzero,\rowthree)
            {\cite[Thm~1.6]{Viola12}};
            \node[rounded corners,black,left=1mm,fill=white,opacity=0.9] at  (axis cs:0.3,\rowthree)
            {$\Delta = 2^{-O(d)} - O(1/n)$};
        \end{axis}
    \end{tikzpicture}
    \caption{The table above depicts the implications of \Cref{thm:main} and \cite[Theorem~1.6]{Viola12}
        for sampling $\bm{U}_k^n$ for different $k$. The plot above it illustrates the size of the corresponding set of bitstrings in the boolean cube.}
    \label{fig:results-summary}
\end{figure}
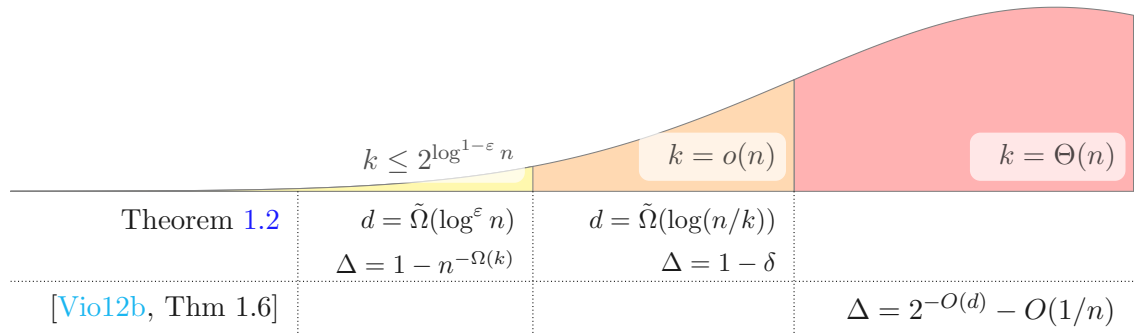

Here are some important challenges that are left open:
\begin{itemize}
    \item Give any non-trivial decision depth (or even locality!) lower bound for linear Hamming weight in the constant-error regime. 

\item Tighten up the decision depth lower bound for $\bm{U}_1^n$ to $\Omega(\log n)$. This would immediately yield tight (up to a constant) decision depth lower bounds for all polynomial $k$. 
\item  Give any locality lower bound for $(\bm{X}, f(\bm{X}))$, where $\bm{X} \sim \{0,1\}^n$ and $f = [(x_1 + \dots + x_n) \bmod p \ge p/2] \oplus x_1 \oplus \dots \oplus x_n$, as this would separate quantum and classical $\NC^0$ circuits for sampling, improving on the partial separation of~\cite{BPP23}. 
\item Find any non-trivial lower bounds for sampling a uniform vector with Hamming weight divisible by $k$, for any $k > 2$. This is likely to give insights on the $\QNC^0$ versus $\NC^0$ separation for sampling.
\item Determine the optimal depth of a switching network which samples a uniform permutation or a uniform vector in a slice.
\end{itemize}

\section{Notation and Tools}
We use boldface letters for random variables, e.g.\ $\bm{a}, \bm{A}, \bm{b}, \bm{B}$. We write $\bm{a}
\sim A$ to say that $\bm{a}$ is distributed according to a distribution $A$, or if $A$ is a set,
according to the uniform distribution over it. For $x \in \{0,1\}^n$, we denote its Hamming weight by
$w(x) = |x| \coloneqq \{i \in [n] \mid x_i = 1\}$. We denote $e_i = 0^{i-1}10^{n-i} \in \{0,1\}^n$. For a
string $x\in\{0,1\}^n$ and a set $T \subseteq [n]$, we write $x_T \coloneqq (x_i)_{i\in T} \in
\{0,1\}^T$. We write $\bm{U}^n_k$ for the uniform distribution of Hamming weight $k$ vectors.  

For two distributions $S$ and $T$ over the same domain $\mathcal{X}$, the \emph{statistical distance} is
defined as
\[ \Delta(S,T) \coloneqq \max_{A \subseteq \mathcal{X}} \left|\Pr_{\bm{x} \sim S}[\bm{x} \in A] -
  \Pr_{\bm{x} \sim T}[\bm{x} \in A]\right| = \frac12 \sum_{a \in \mathcal{X}} \left|\Pr_{\bm{x} \sim
      S}[\bm{x} = a] - \Pr_{\bm{x} \sim T}[\bm{x} = a]\right|.
\]
The statistical distance between two random variables is the statistical distance between their
distributions. We say that the distribution $S$ is \emph{$\epsilon$-close} from the distribution $T$ if
$\Delta(S,T) < \epsilon$. Otherwise, the distributions are \emph{$\epsilon$-far}.

A function $f\colon \{0,1\}^m \to \{0,1\}^n$ is \emph{$d$-local} if each of its output bits depends only
on $d$ input bits. A function $f\colon \{0,1\}^m \to \{0,1\}^n$ has \emph{decision depth} $d$ if each of
its output bits can be computed as a depth-$d$ decision tree, i.e.\ decided with at most $d$ adaptive
input bit queries. If a function is $d$-local, then it has decision depth at most $d$. 

We are going to use a very simple form of the FKG inequality:
\begin{theorem}[\cite{Harris60,FKG71}]
    \label{thm:FKG}
    Suppose that $X$ is a product distribution over $\{0,1\}^n$ (that is, $\Pr[X = x] = \prod_{i \in [n]}
    \Pr[X_i = x_i]$). Let $A, B \subseteq \{0,1\}^n$ be two monotone events (if $x \in A$ and $x_i \leq
    y_i$ for all $i \in [n]$ then $y \in A$, and similarly for $B$). Then
    \[
        \Pr_{\bm{x} \sim X}[\bm{x} \in A \cap B] \ge
        \Pr_{\bm{x} \sim X}[\bm{x} \in A] \cdot \Pr_{\bm{x} \sim X}[\bm{x} \in B]. \qedhere
    \]
\end{theorem}

\subsection{Robust Sunflowers}
In this section, we discuss robust sunflowers.


\begin{definition}[Robust sunflower]
    Let $0 < \alpha, \beta < 1$ be parameters, let $\mathcal{F}$ be a set system, and let $K \coloneqq
    \bigcap_{S \in \mathcal{F}} S$ be the intersection of all sets in $\mathcal{F}$, which we refer to as
    the \emph{kernel}. The family $\mathcal{F}$ is an \emph{$(\alpha, \beta)$-robust sunflower} if
    \begin{enumerate}
        \item $K \not\in \mathcal{F}$;
        \item $\Pr_{\bm{R}} \left[\exists S \in \mathcal{F}\colon S \subseteq \bm{R} \cup K\right] \ge 1
            - \beta$, where each element of the universe appears in $\bm{R}$ with probability $\alpha$
            independently.

            We can write this condition in the equivalent form $\Pr_{\bf{R}}[\exists S \in
            \mathcal{F}\colon \mathbf{R} \supseteq S \mid \mathbf{R} \supseteq K] \ge 1 - \beta$.
    \end{enumerate}
    A set system is called \emph{$(\alpha, \beta)$-satisfying} if it is an $(\alpha, \beta)$-robust
    sunflower with an empty kernel. 
\end{definition}

Large enough set systems always contain a robust sunflower, as proved by Rossman~\cite{Rossman14} and
improved by later authors. 

\begin{theorem}[\cite{ALWZ21, BCW21, Rao19}]
    \label{thm:improved-sunflower}
    There exists a constant $B > 0$ such that the following holds for all $p, \epsilon \in (0, 1/2]$. Let
    $\mathcal{F}$ be a family of sets of size exactly $d$ such that $|\mathcal{F}| \ge (B \log
    (d/\epsilon)/p)^d$. Then $\mathcal{F}$ contains a $(p, \epsilon)$-robust sunflower.
\end{theorem}

\begin{corollary}
    \label{cor:improved-sunflower-cor}
    There exists a constant $B > 0$ such that the following holds for all $p, \epsilon \in (0, 1/2]$. Let
    $\mathcal{F}$ be a family of non-empty sets of size \emph{at most} $d$ such that $|\mathcal{F}| \ge d
    \cdot (B \log (d/\epsilon)/p)^d$. Then $\mathcal{F}$ contains a $(p, \epsilon)$-robust sunflower.
\end{corollary}

\begin{proof}
    Let $d_0 \in [d]$ be the most common size of sets in $\mathcal{F}$. Then the number of sets of size
    $d_0$ is at least $(B \log (d/\epsilon)/p)^d$, which allows us to apply \Cref{thm:improved-sunflower}
    to these sets.
\end{proof}


If we remove a single petal from a robust sunflower, then it remains a robust sunflower (with slightly worse parameters).

\begin{lemma}
\label{lem:robust-sunflower-petal-removing}
Suppose that $N_1, \dots, N_k \subseteq X$ is a $(p, \epsilon)$-robust sunflower with kernel $K$. Then for every $i$, the sets $N_1, \dots, N_{i-1}, N_{i+1}, \dots, N_k$ form a $(2p, 2\epsilon)$-robust sunflower with kernel $K$.
\end{lemma}
    
\begin{proof}
    Let $\bm{\tau}$ be distributed over $[3]^{X \setminus K}$ such that $\Pr[\bm{\tau}_i = 1] =
    \Pr[\bm{\tau}_i = 2] = p$, $\Pr[\bm{\tau}_i = 3] = 1 - 2p$, and the coordinates of $\bm{\tau}$ are
    independent. For $\ell \in [3]$, let $\bm{\tau}^\ell = \{j \in X \mid \bm{\tau}_j = \ell \} \cup
    K$. The definition of $(p, \epsilon)$-robust sunflower implies that for every $\ell \in [2]$ we have
    $\Pr[\exists j \in [k]\colon \bm{\tau}^\ell \supseteq N_j] \ge 1 - \epsilon$. An application of the
    union bound implies that
    \[
        \Pr\left[\exists j \in [k]\colon \bm{\tau}^1 \supseteq N_j\land \exists j' \in [k]\colon
          \bm{\tau}^2 \supseteq N_{j'}\right] \ge 1 - 2\epsilon.
    \]
    If $j = j'$, then since $\bm{\tau}^1 \cap \bm{\tau}^2=K$, we have $N_j = K$, which is impossible by
    the definition of a sunflower. Thus $j \neq j'$ whenever the event happens. Let $\bm{R}$ be a
    distribution of subsets of $X$ where each element appears in $\bm{R}$ independently with probability
    $2p$. Then since $(\bm{R} | \bm{R} \supseteq K)$ has the same distribution as $\bm{\tau}^1 \cup
    \bm{\tau}^2$, we have
    \[
        \Pr_{\bm{R}}\left[\exists j \neq j' \in [k]\colon \bm{R} \supseteq N_j \land \bm{R} \supseteq
          N_{j'} \,\middle|\, \bm{R} \supseteq K\right] \ge 1 - 2\epsilon.
    \]
    In particular, for every $i \in [k]$ we have
    \[
        \Pr_{\bm{R}}\left[\exists j \neq i \in [k]\colon \bm{R} \supseteq N_j \,\middle|\, \bm{R}
          \supseteq K\right] \geq 1 - 2\epsilon,
    \]
    and so the sets $N_1, \dots, N_{i - 1},N_{i + 1}, N_k$ form a $(2p, 2\epsilon)$-robust sunflower with
    kernel $K$.
%
\end{proof}

Another lemma we use is very similar to the standard connection between robust sunflowers and the classical ones (see e.g.\ Lemma 1.6 in \cite{ALWZ21}):
\begin{lemma}
\label{lem:k-petals-in-robust-sunflower}
    Suppose that $N_1, \dots, N_m \subseteq X$ is a $(1/(2k), \epsilon)$-robust sunflower with a kernel
    $K$. Then
    \[
        \Pr_{\bm{R} \sim 2^X}\left[\exists I\in \binom{[m]}{k} \forall i \in I\colon \bm{R} \supseteq N_i
          \middle| \bm{R} \supseteq K\right] \ge 1 - \epsilon k.
    \]
\end{lemma}
\begin{proof}
    Let $\tau \sim [2k]^{X \setminus K}$, and $\tau^i \coloneqq \{j \in X \setminus K \mid \tau_j = i\}
    \cup K$ for $i \in [2k]$. Then $\bm{\tau}^1 \cup \dots \cup \bm{\tau}^k$ is distributed equivalently
    to $(\bm{R} \mid \bm{R} \supseteq K)$ where $\bm{R} \sim 2^X$. On the other hand, by the definition
    of the $(1/(2k), \epsilon)$-robust sunflower, for each $i \in [2k]$ we get
    \[
        \Pr[\exists j \in [m]\colon \bm{\tau}^i \supseteq N_j] \ge 1 - \epsilon.
    \]
    Since $\bm{\tau}^i \cap \bm{\tau}^{i'} = K$ for any $i \neq i' \in [2k]$, the lemma follows by the union bound over $i \in [k]$.
\end{proof}

\section{Sampling Uniform Hamming Weight \texorpdfstring{$k$}{k} Distributions}
In this section we prove the following results mentioned in the introduction, which we restate here for convenience.

\MainThm*

\MainThmUone*

We first prove \Cref{thm:U1-main}, in \Cref{sec:proof-of-U1main}. We then prove \Cref{item:subpoly} of \Cref{thm:main} in \Cref{sec:subpoly}, and \Cref{item:sublinear} of the \namecref{thm:main} in \Cref{sec:sublinear}. We prove the ``moreover'' part in \Cref{sec:moreover-of-thm-main}.


\subsection{Proof of \texorpdfstring{\Cref{thm:U1-main}}{Theorem \ref{thm:U1-main}}}
\label{sec:proof-of-U1main}
We prove a more general result which immediately yields \Cref{thm:U1-main}.

First let us sketch a proof of \Cref{thm:U1-main} for $d$-local functions. Suppose that $\Delta(\bm{U}_1^n,\bm{X}) \leq 1-\eta$. Call a coordinate $i$ \emph{good} if $\Pr[\bm{X} = e_i] \geq 1/n^2$. Since $\Pr[U_1 = e_i] = 1/n$, many coordinates are good: at least $\Omega(\eta n)$.

Let $\bm{Y} \sim \{0,1\}^m$ denote the random input bits. Each $X_i$ depends on some subset $N_i
\subseteq [m]$ of coordinates of size at most $d = \tau \log n/\log\log n$, say $\bm{X}_i =
f_i(\bm{Y}_{N_i})$. 

For each good coordinate $i$, we choose an assignment $\alpha_i \in f_i^{-1}(1)$ which maximizes the
conditional probability $\Pr[\bm{X} = e_i \mid \bm{Y}_{N_i} = \alpha_i]$, that is, the probability that
if $\bm{Y}_{N_i} = \alpha_i$ then all other output bits are $0$. This probability is at least $1 / (2^dn^2)
= \Omega(1 / n^3)$.

The assignments $\alpha_i$ do not necessarily agree with each other. However, a random assignment $\rho$
to $\bm{Y}$ agrees with at least $\Omega(\eta n/2^d) = \Omega(\eta n^{1 - o(1)})$ of them. Let $T$
consists of the domains of the assignments $\alpha_i$ which agree with $\rho$. These domains are distinct
since $N_i = N_j$ implies $\alpha_i = \alpha_j$ and hence
that $\Pr[\bm{X} = e_i \mid \bm{Y}_{N_i} = \alpha_i] = 0$. The choice of $d$
guarantees that $T$ supports a $(1/4,\epsilon)$-robust sunflower $\mathcal S$, for any $\epsilon$ which
is inverse-polynomial in $n$. Let $K$ be the kernel of $\mathcal S$.

If we remove any single petal $i$ from $\mathcal S$ then
by \Cref{lem:robust-sunflower-petal-removing}
the result is a $(1/2,2\epsilon)$-robust
sunflower, and so given that $\bm{Y}_K$ agrees with $\rho$, the probability that $\bm{X}_j = 1$ for some
$j \neq i$ is at least $1-2\epsilon$. If we replace the condition with ``$\bm{Y}_{N_i}$ agrees with $\rho$'' (and
so with $\alpha_i$),
then intuitively, the probability can only increase, and this can be formalized
using the FKG inequality (\Cref{thm:FKG}). By definition of $\alpha_i$, this means that $\Pr[\bm{X} =
e_i] \leq |f_i^{-1}(1)| 2 \epsilon \leq 2^{d+1} \epsilon$. Choosing $\epsilon = 1/2^{d+1}n^2$ shows that
$i$ is not good, and we reach a contradiction.

\medskip

We move on to prove the generalization of \Cref{thm:U1-main}.

\begin{theorem}
    \label{thm:tildelog-U1-general}
    Let $\bm{Y} \sim \{0,1\}^m$ be the input bits of the $n$-bit source $\bm{X}$. Suppose that every bit
    of $\bm{X}$ is computed as a DNF of bits of $\bm{Y}$ of size at most $s$ and width at most $d$. For
    every $\kappa \in \mathbb{R}$ there exists a constant $\tau$ such that for $d = \tau \log n / \log \log n$ and $s \le
    \kappa n^\kappa$, we have $\Delta(\bm{X}, \bm{U}_1^n) = 1 - \eta = 1 - n^{-\Omega(1)}$.
\end{theorem}

This implies \Cref{thm:U1-main} since the output of a decision tree of depth $d$ can be represented as a
DNF of size at most $2^d$ and width at most $d$. In our case $d = o(\log n)$ and so $2^d \leq n$ (for
large enough $n$). 

\begin{proof}
    We say that an output bit $i \in [n]$ is \emph{good} if $\Pr[\bm{X} = e_i] \ge 1/n^2$. Let $G
    \subseteq [n]$ be the set of all good bits, and let $\overline{G} = [n] \setminus G$. Let us estimate the size of $G$: $\Pr[\bm{X} \in \{e_i
    \mid i \not\in G\}] \le |\overline{G}|/n^2$, but $\Pr[\bm{U}_1^n \in \{e_i \mid i \notin G\}] =
    |\overline{G}|/n)$, so $|\overline{G}| \cdot (1/n - 1/n^2) \le 1-\eta$, which yields $|G| =
    \Omega(\eta n)$. 

    For each $i \in G$, since each bit of $\bm{X}$ is represented as a DNF we have
    \( \bm{X}_i = \bigvee_{j \in [s_i]} [\bm{Y}_{N_i^j} = \alpha_i^j] \),
    where $N_i^1, \dots, N_i^{s_i} \subseteq [m]$ are sets, $\alpha_i^j \in \{0,1\}^{N_i^j}$ are truth assignments, and $s_i \le s$. 
    By the law of total probability we have
    \[\Pr[\bm{X} = e_i] \le \sum_{j \in [s_i]} \Pr[\bm{X} = e_i \land \bm{Y}_{N_i^j} = \alpha_i^j] \le s \Pr[\bm{X} = e_i \land \bm{Y}_{N_i^{\max}} = \alpha_i^{\max}],\]
    where $N_i^{\max}$ and $\alpha_i^{\max}$ correspond to the term in the DNF maximizing the probability $\Pr[\bm{X} = e_i \land \bm{Y}_{N_i^j} = \alpha_i^j]$. 

    Consider the expected number of good output bits such that $\bm{Y}_{N_i^{\max}} = \alpha_i^{\max}$: 
    \[\mathbb{E}\left[\sum_{i \in G} [\bm{Y}_{N_i^{\max}} = \alpha_i^{\max}] \right] \ge |G| 2^{-d}.\]
    Hence there exists an assignment $\rho$ to the input bits such that for at least $|G| 2^{-d}$ good
    output bits, we have $\rho_{N_i^{\max}} = \alpha^{\max}_i$. Let $T \subseteq G$ be the set of those
    output bits. If $i,j \in T$ then $N_i^{\max} \neq N_j^{\max}$, since otherwise $\bm{Y}_{N_i^{\max}} =
    \alpha_i^{\max}$ implies that also $\bm{Y}_{N_j^{\max}} = \alpha_j^{\max}$ and so $\bm{X} \neq e_i$,
    and so $\Pr[\bm{X} = e_i] = 0$, contradicting $i \in G$. Observe moreover that none of the sets $N_i$
    for $i \in T$ is empty, since otherwise $|T|=1$ and we get an immediate contradiction with the size
    of $G$ for any $d = o(\log n)$.

    \paragraph{Case 1.} $|T| < d (4 B \log (d / \epsilon))^d$. In this case, we immediately get the lower
    bound on $\delta$. Indeed, the inequality $|G| 2^{-d} \le |T| < d (4 B \log (d / \epsilon))^d$
    implies $|G| \le d (8 B \log (d/\epsilon))^d$, which together with $|G| = \Omega(\eta n)$ yields
    $\eta \le d (8 B \log (d/\epsilon))^d / n$. If $\epsilon$ is inverse polynomial in $n$, then for
    small enough $\tau$ we get $\eta = n^{-\Omega(1)}$ with $d = \tau \log n/\log \log n$.

    \paragraph{Case 2.} $|T| \ge  d (4 B \log (d/\epsilon))^d$. Then by \Cref{cor:improved-sunflower-cor}
    there exists a $(1/4, \epsilon)$-robust sunflower formed by the sets $N_{t_1}^{\max}, \dots,
    N_{t_k}^{\max}$ for $\{t_1, \dots, t_k\} \subseteq T$ (recall the sets $N_i^{\max}$ for $i \in T$ are
    all distinct, and none of them is empty). Let $K$ denote the kernel of this sunflower. Consider an
    arbitrary petal $t_i$ of this sunflower. By \Cref{lem:robust-sunflower-petal-removing} we have that
    $\{N_i^{\max}\}_{i \in T \setminus \{t_i\}}$ is a $(1/2, 2\epsilon)$-robust sunflower. Let $\bm{U}$
    be the set of indices such that $\bm{Y}_k = \rho_k$. Then
    \begin{align*}
      \Pr[\bm{X}_j = 1 \text{ for some } j \in T \setminus \{t_i\} \mid \bm{Y}_{N_{t_i}^{\max}} =
      \rho_{N_{t_i}^{\max}}] &\geq\\
      \Pr[\bm{U} \supseteq N_j^{\max} \text{ for some } j \in T \setminus \{t_i\} \mid \bm{U} \supseteq
      N_{t_i}^{\max}] &= \\
      \frac{\Pr[\bm{U} \supseteq N_{t_i}^{\max} \text{ and } \bm{U} \supseteq N_j^{\max} \text{ for some
      } j \in T \setminus \{t_i\} \mid \bm{U} \supseteq K]}{\Pr[\bm{U} \supseteq N_{t_i}^{\max} \mid
      \bm{U} \supseteq K]} &\geq & \text{\Cref{thm:FKG}}\\
      \Pr[\bm{U} \supseteq N_j^{\max} \text{ for some } j \in T \setminus \{t_i\} \mid \bm{U} \supseteq
      K] &\geq \\
      1 - 2\epsilon,
    \end{align*}
    where the last inequality is due to the definition of a $(1 / 2, 2 \epsilon)$-robust
    sunflower. Recall that by the choice of $T$, we have $\rho_{N_{t_i}^{\max}} =
    \alpha_{t_i}^{\max}$. Therefore
    \[
        \Pr[\bm{X} \neq e_{t_i} \mid \bm{Y}_{N_{t_i}^{\max}} = \alpha_{t_i}^{\max}] = \Pr[\bm{X}_j = 1
        \text{ for some } j \in T \setminus \{t_i\} \mid \bm{Y}_{N_{t_i}^{\max}} = \rho_{N_{t_i}^{\max}}]
        \ge 1 - 2\epsilon.
    \]
    Thus $\Pr[\bm{X}=e_{t_i} \land \bm{Y}^{N_{t_i}^{\max}} = \alpha^{\max}_{t_i}] \le 2\epsilon$. By the
    choice of $\alpha^{\max}_{t_i}$, $\Pr[\bm{X}=e_{t_i}] \le s \cdot 2\epsilon$.
    Picking $\epsilon < 1 / (2 s n^2)$, which is inverse polynomial in $n$ as required in Case~1, we get
    that $\Pr[\bm{X}=e_{t_i}] < 1/n^2$, so $t_i$ is bad, which contradicts the choice of $T$.
\end{proof}

\subsection{A Generalized Version of \texorpdfstring{\Cref{thm:tildelog-U1-general}}{Theorem
        \ref{thm:tildelog-U1-general}}}
In this section, we generalize \Cref{thm:tildelog-U1-general} so it can be used to prove
\Cref{item:sublinear} of \Cref{thm:main}. The proof follows the same path as the proof of
\Cref{thm:tildelog-U1-general}, we decided to include both proofs for simplicity. 

\begin{theorem}
\label{thm:general-sunflowers}
     Let $\bm{Y} \sim \{0,1\}^m$ be the input bits of the $n$-bit source $\bm{X}$. Suppose that every bit
     of $\bm{X}$ is computed as a DNF of bits of $\bm{Y}$ of size at most $s$ and width at most $d$. Let
     $t \in [n]$ be a parameter, $\alpha(n)$ be a function and let $\bm{F}$ be a distribution over
     $\{0,1\}^n$ with the following properties:
       \begin{itemize}
           \item For every set $T \subseteq [n]$ such that $T \ge n/2$ we have $\Pr[\bm{F}_T = 0^T] \le \alpha(n)$;
           \item $\Pr[|\bm{F}| > t] \le \alpha(n)$.
       \end{itemize}
    If $2 d \cdot t \cdot (40 B t \log n)^d \le n$ then $\Delta(\bm{X}, \bm{F}) \ge 1 - 2\alpha(n) -
    1/2n$. Here $B$ is the constant from \Cref{cor:improved-sunflower-cor}.
\end{theorem}
\begin{proof}
    We say that an output bit $i \in [n]$ is \emph{good} if $\Pr[\bm{X}_i = 1 \land |\bm{X}| \le t] \ge
    1/n^2$. Let $G$ be the set of good bits and let $\overline{G} \coloneqq [n] \setminus G$. Suppose
    that $|G| \le n/2$. Then by the conditions on $\bm{F}$ we have $\Pr[\bm{F}_{\overline{G}} =
    0^{\overline{G}}] \le \alpha(n)$. On the other hand $\Pr[\bm{X}_{\overline{G}} \neq 0^{\overline{G}}
    \land |\bm{X}| \le t] < |\overline{G}|/n^2 < 1/2n$. Then $\Delta(\bm{X}, \bm{F}) \ge 1-2\alpha(n)
    -1/2n$, as required. In the rest of the proof, we derive a contradiction with $|G| \ge n/2$. 

    As in the proof of \Cref{thm:tildelog-U1-general}, we pick the likeliest term in the DNF
    representation of each of the output bits. For each $i \in G$, since each bit of $\bm{X}$ is
    represented as a DNF, we have 
    \( \bm{X}_i = \bigvee_{j \in [s_i]} [\bm{Y}_{N_i^j} = \alpha_i^j] \),
    where $N_i^1, \dots, N_i^{s_i} \subseteq [m]$ are sets, $\alpha_i^j \in \{0,1\}^{N_i^j}$ are truth assignments, and $s_i \le s$. 
    By the law of total probability, we have
    \[\Pr[\bm{X}_i = 1 \land |\bm{X}| \le t] \le \sum_{j \in [s_i]} \Pr[|\bm{X}| \le t \land
        \bm{Y}_{N_i^j} = \alpha_i^j] \le s \Pr[|\bm{X}| \le t \land \bm{Y}_{N_i^{\max}} =
        \alpha_i^{\max}],\]
    where $N_i^{\max}$ and $\alpha_i^{\max}$ correspond to the term in the DNF maximizing the probability
    $\Pr[|\bm{X}| \le t \land \bm{Y}_{N_i^j} = \alpha_i^j]$. 
    The expected number of good output bits with $\bm{Y}_{N_i^{\max}} = \alpha_i^{\max}$ is at least
    $2^{-d} |G|$, so there exists an assignment $\rho$ to the input bits such that $\rho_{N_i^{\max}} =
    \alpha_i^{\max}$ for at least $|G| 2^{-d}$ good output bits. Let $T \subseteq G$ be the set of these
    output bits.
    
    Let us estimate how many distinct elements are in the set $\mathcal{N} \coloneqq \{N_i^{\max} \mid i
    \in T\}$. Suppose there exist $i_1, \dots, i_{t+1} \in T$ such that $N_{i_1}^{\max} = \dots =
    N_{i_{t+1}}^{\max}$. Then, by the definition of $\rho$, we have $\alpha_{i_1}^{\max} = \dots =
    \alpha_{i_{t+1}}^{\max}$ as well. Thus $\bm{Y}_{N_{i_1}^{\max}} = \alpha_{i_1}^{\max}$ implies that
    for every $j \in [t+1]$ we have $\bm{Y}_{N_{i_j}^{\max}} = \alpha_{i_j}^{\max}$, which in turn
    implies that $|\bm{X}| \ge t+1$, and so $\Pr[\bm{Y}_{N_{i_1}^{\max}} = \alpha_{i_1}^{\max} \land
    |\bm{X}| \le t] = 0$, which contradicts that $i_1$ is good. Hence $|\mathcal{N}| \ge |T|/t \ge |G|
    2^{-d}/t \ge 2^{-d}n/2t$.

    Let $\epsilon > n^{-5}$ be a parameter to be chosen later. By the condition on $n$ we have
    $|\mathcal{N}| \ge 2^{-d}n/2t \ge d \cdot (4 B t \log (d/\epsilon))^d$, and so $\mathcal{N}$ contains
    a $(1/(4t), \epsilon)$-robust sunflower $\mathcal{S}$. Let $K$ denote the kernel of this sunflower.
    
    Fix an arbitrary petal $p \in \mathcal{N}$. Then $\mathcal{N} \setminus \{p\}$ is a
    $(1/(2t),2\epsilon)$-robust sunflower by \Cref{lem:robust-sunflower-petal-removing}. Now by
    \Cref{lem:k-petals-in-robust-sunflower} we have
    \[\Pr_{\bm{R} \sim 2^{[m]}} \left[\text{There are }t\text{ distinct petals of }\mathcal{N}\setminus
          \{p\}\text{ contained in }\bm{R} \mid \bm{R} \supseteq K\right] \ge 1-2t\epsilon.\]
    Let $P \subseteq T$ be the indices of the output bits corresponding to the elements of
    $\mathcal{N}\setminus \{p\}$, let $i$ be the index of the output bit corresponding to the petal $p$,
    and let $\bm{U}$ be the set of indices of input bits such that $\bm{Y}_t = \rho_t$. Then 
    \begin{align*}
      \Pr[|\bm{X}| > t \mid \bm{Y}_{N_i^{\max}} = \rho_{N_i^{\max}}] &=\\
      \Pr[|\bm{X}_{[n]\setminus \{i\}}| \ge t \mid \bm{Y}_{N_i^{\max}} = \rho_{N_i^{\max}}]& \ge \\
      \Pr\left[\sum\limits_{j \in P} [\bm{U} \supseteq N_j^{\max}] \ge t \,\middle|\, \bm{U} \supseteq
        N_i^{\max}\right] &= \\
      \frac{ \Pr[\sum_{j \in P \cup \{i\} } [\bm{U} \supseteq N_j^{\max}] \ge t \mid
          \bm{U} \supseteq K]}{\Pr[\bm{U} \supseteq N_i^{\max} \mid \bm{U} \supseteq K]} &\ge \text{
                                                                                           \Cref{thm:FKG}}\\
      \Pr\left[\sum_{j \in P} [\bm{U} \supseteq N_j^{\max}] \ge t \,\middle|\, \bm{U} \supseteq
      K\right]&\ge \\
      1 - 2t\epsilon.
    \end{align*}
    Recall that by the choice of $T$ we have $\rho_{N_i^{\max}} = \alpha_i^{\max}$, hence
    $\Pr[|\bm{X}| > t \mid \bm{Y}_{N_i^{\max}} = \alpha_i^{\max}] \ge 1 - 2t\epsilon$.
    Thus \[\Pr[|\bm{X}|\le t \land \bm{X}_i = 1] \le s \cdot \Pr[|\bm{X}| \le t \land \bm{Y}_{N_i^{\max}}
        = \alpha_i^{\max}] \le 2st \cdot \epsilon.\] 
    Picking $\epsilon = 1/(4 st n^2) \ge n^{-5}$, we get a contradiction with $i$ being good. 
\end{proof}

\subsection{Subpolynomial Weights}
\label{sec:subpoly}
Although \Cref{thm:general-sunflowers} implies \Cref{item:subpoly} of \Cref{thm:main}, we give a simpler
proof via a reduction from $\bm{U}_1^n$. 
\begin{lemma}
\label{lem:reduction-general}
    Let $S \subseteq \{0,1\}^n$ and let $\bm{S} \sim S$. Suppose that $\bm{S}$ can be sampled with a
    depth-$d$ decision forest with error $\eta$. Assume furthermore that for each $s \in S$ there exists
    a decision tree $T_s$ of depth $k$ that accepts $s$ and does not accept any of $S \setminus
    \{s\}$. Then there exists a decision depth-$kd$ sampler for $\bm{U}^{|S|}_1$ with error $\eta$.
\end{lemma}
\begin{proof}
    Let $\bm{Y}$ be the distribution sampled by the sampler for $\bm{S}$. For each output bit of our
    sampler for $\bm{U}^{|S|}_1$ we take a unique element $s \in S$ and implement each of the queries of
    $T_s$ via the query to the bits of $\bm{Y}$ (which makes at most $d$ queries to the input bits). This
    results in a $kd$-deep decision tree $T'_s$. Let $\bm{X}$ be the sampled distribution. Then
    \begin{align*}
      \Delta(\bm{X}, \bm{U}^{|S|}_1) &=\frac{1}{2}\left(\Pr[w(\bm{X}) \neq 1] + \sum_{s\in S}
                                       \bigl|\Pr[\bm{X} = e_s] - 1/|S|\bigr|\right)\\
                                     &=\frac{1}{2}\left(\Pr[\bm{Y}\not\in S] + \sum_{s\in S}
                                       \bigl|\Pr[\bm{Y} = s] - 1/|S|\bigr|\right)=\Delta(\bm{Y}, \bm{S})
                                       = \eta.
    \end{align*}
\end{proof}

\begin{restatable}{corollary}{Ukmain}
\label{cor:Uk-main}
    For some constant $\tau' > 0$ and every $\epsilon \in (0,1)$ every $(\tau' \log^{\epsilon}
    n)$-decision depth sampler outputs a distribution $(1-n^{-\Omega(k)})$-far from $\bm{U}^n_k$ for $k
    \in [\log n, 2^{\log^{1-\epsilon} n}]$. If $k < \log_2 n$, this holds for every $(\tau'
    \log^{\epsilon} n/\log \log n)$-local sampler.
\end{restatable}
\begin{proof}
    The decision tree that queries all the elements of a $k$-size set and accepts iff all of them are $1$
    satisfies the condition in \Cref{lem:reduction-general}. If we had a $(1-\delta)$-error sampler for
    $\bm{U}^n_k$, we would get a $(1-\delta)$-error sampler for $\bm{U}^{\binom{n}{k}}_1$ with decision
    depth $kd$, which by \Cref{thm:U1-main} yields that $\delta=\binom{n}{k}^{-\Omega(1)}=n^{-\Omega(k)}$
    whenever $kd \le \tau \log \binom{n}{k} / \log \log \binom{n}{k}$. Since $\log \binom{n}{k} =
    \Theta(k \log (n/k))$, we get $d = \Omega(\log (n/k) / (\log k + \log \log n))$.
\end{proof}

\subsection{Sublinear Weights}
\label{sec:sublinear}
In this section, we prove \Cref{item:sublinear} of \Cref{thm:main}.
\begin{lemma}
    Suppose $\bm{X}$ is sampled with decision depth $d$. Then for every small enough $\epsilon > 0$ there
    exists a constant $\tau$ such that if $d \le \tau \log(n/k) / \log\log (n/k)$ then $\Delta(\bm{X},
    \bm{U}_k^n) \ge 1 - 2\epsilon - \frac{1}{2n}$.
\end{lemma}
\begin{proof}
    Consider the first $\ell \coloneqq \epsilon^{-1} \cdot n/k$ bits of the sampler: $\bm{X}_{\le \ell}
    \coloneqq \bm{X}_1,\dots,\bm{X}_{\ell}$. Let $\bm{Y}$ be the first $\ell$ bits of the distribution
    $\bm{U}_k^n$. We show that $\bm{Y}$ satisfies the conditions of \Cref{thm:general-sunflowers} for
    $t=\epsilon^{-2}$ and $\alpha(n)=\epsilon$. First, we have
    \[\Pr[|\bm{Y}| > t] = \Pr\left[|\bm{Y}| > \epsilon^{-1} \E[|\bm{Y}|]\right] < \epsilon.\]
    Now let $T$ be any subset of $[\ell]$ of size at least $\ell/2$. Then
    \[
        \Pr[\bm{Y}_T = 0^T] = \binom{n - \ell / 2}{k} \binom{n}{k}^{-1} = \prod_{i = 0}^{k - 1}
        \frac{n - i - \ell / 2}{n - i} \le \left(1 - \frac{\ell}{2n} \right)^k =
        \left( 1 - \frac{1}{2\epsilon k} \right)^k \le
        e^{-2\epsilon^{-1}} < \epsilon.\]
    Here we assumed $k < n/2$, since otherwise the lemma is trivially true.

    Applying \Cref{thm:general-sunflowers}, for small enough $\tau$ (which depends on $\epsilon$) we have
    $\Delta(\bm{X}_{\le \ell}, \bm{Y}) = 1 - 2\epsilon - 1/2n$. To finish the proof, observe that
    $\Delta(\bm{X}, \bm{U}_k^n) \ge \Delta(\bm{X}_{\le \ell}, \bm{Y})$, since the random variables on the
    RHS are the marginals of the variables on the LHS. 
\end{proof}

\subsection{Unions of Slices}
\label{sec:moreover-of-thm-main}
In this section, we prove the ``moreover'' part of \Cref{thm:main} by observing that the distribution $\bm{U}_S^n$ is close to $\bm{U}_{\max_{x \in S} x}^n$ as long as $\max_{x \in S} x = o(n)$. 
\begin{proposition}
    Let $k=o(n)$ and suppose that $S \subseteq \{0,1,\dots, k\}$ with $k \in S$. Then we have $\Delta(\bm{U}_k^n,\bm{U}_S^n) = o(1)$.
\end{proposition}
\begin{proof}
    We use the notation $\binom{n}{S} \coloneqq \sum_{i\in S} \binom{n}{i}$.
    \begin{align*}
        \Delta(\bm{U}_k^n,\bm{U}_{S}^n) &= \frac12 \binom{n}{S \setminus \{k\}}\binom{n}{S}^{-1} + \frac12 \left(1 - \binom{n}{k} \binom{n}{S}^{-1}\right)\\
        &= \binom{n}{S \setminus \{k\}} \binom{n}{S}^{-1}\\
        &\le \binom{n}{k}^{-1} \sum_{i \in S \setminus \{k\}} \binom{n}{i}\\
        &\le \sum_{i = 0}^{k-1} \left(\frac{k}{n-i}\right)^{i-k} = \Theta(k/n). \qedhere
    \end{align*}  
\end{proof}

The case of $S$ with $n - \min_{x \in S} x = o(n)$ reduced to the case where $\max_{x \in S} x = o(n)$ by observing that flipping all output bits can be done with no increase in the decision depth of the sampler.

\section{Local Certificates}

In this section, we explore the power of local proof systems. \Cref{sec:ECC} gives an example of a language that requires locality $\Omega(n)$, which is inspired by a similar lower bound in the context of sampling~\cite{LV12}. \Cref{sec:majority} then gives our main result, a lower bound on the locality of proof systems for $\Maj^{-1}(1)$.

\smallskip

 We say that an input bit $i \in [m]$ \emph{affects} an output bit $j\in [n]$, or equivalently that the output bit $j$ \emph{depends} on the input bit $i$, if there exist inputs $x,x' \in \{0,1\}^m$, differing only in the $i$th bit, such that $f(x)_j \neq f(x')_j$.  A function $f\colon\{0,1\}^m \to \{0,1\}^n$ is \emph{$d$-local} if each of its output bits depends on at most $d$ input bits. 

\subsection{Error-correcting codes} \label{sec:ECC}
In this section, we show that a good error-correcting code requires a proof system of a linear locality. This showcases the simple counting technique that we also use for our majority lower bound.  

\begin{proposition}
    Let $C \subseteq \{0,1\}^n$ be a good code, that is, $|C| \ge 2^{\alpha n}$ and for every $x\neq y \in C$, the Hamming distance between $x$ and $y$ is at least $\beta n$, where $\alpha$ and $\beta$ are constants in $(0,1)$.
    
    If $f\colon \{0,1\}^m \to \{0,1\}^n$ is a $d$-local function and $f(\{0,1\}^m) = C$ then $d \ge \alpha \beta n$
\end{proposition}
\begin{proof}
    We may assume w.l.o.g.\ that all input bits of $f$ affect some output bits. 
    Take an arbitrary input bit $i \in [m]$ and an output bit $j \in [n]$ that depends on $i$. Then take $x,x' \in \{0,1\}^m$ that differ only in the $i$th coordinate and such that $f(x)_j \neq f(x')_j$. Since $f(x) \neq f(x') \in C$, they must be at Hamming distance at least $\beta n$, hence $i$ must affect at least $\beta n$ output bits, as $x$ and $x'$ only differ in $i$. Thus every bit affects at least $\beta n$ output bits.
    
    There are at least $\alpha n$ input bits since $|C| = 2^{\alpha n}$. Therefore there are $\alpha \beta n^2$ input-output pairs in which the input bit affects the output bit. On the other hand, there are at most $dn$ such pairs, hence $d \ge \alpha \beta n$.
\end{proof}

\subsection{Majority} \label{sec:majority}
Let $\Maj_n$ be the the set $\{x\in\{0,1\}^n \mid |x| \ge n/2\}$. First, let us give a simple upper bound on locality which is implicit in \cite{KLMS16}.

\begin{proposition}[essentially Theorem~3.9 and Corollary~3.10 in \cite{KLMS16}]
There exists an $O(\log^2 n)$-local function $f\colon \{0,1\}^m \to \{0,1\}^n$ such that $f(\{0,1\}^m) =\Maj_n^{-1}(1)$.
\end{proposition}
\begin{proof}
   For simplicity, suppose that $n$ is odd. Construct a binary tree whose root is the interval $[1,n]$, whose leaves are the singletons $\{1\},\ldots,\{n\}$, and in which each internal node $[\ell,r]$ has two children $[\ell,c],[c+1,r]$, where $c = \lfloor(\ell+r)/2\rfloor$. We can construct such a tree whose depth is $O(\log n)$.
For each interval $[\ell,r]$ in the tree we will have a label $w(\ell,r)$ whose value ranges from $0$ to $r-\ell+1$, which is supposed to indicate $x_\ell+\cdots+x_r$ (where $x_1,\ldots,x_n$ is the output). We implement the variables using $O(\log n)$ input bits.

An internal node $[\ell,r]$ with children $[\ell,c],[c+1,r]$ is \emph{consistent} if $w(\ell,r) = w(\ell,c) + w(c+1,r)$. In addition, if the internal node is the root $[1,n]$, we require $w(1,n) \geq (n+1)/2$.
Each position $i \in \{1,\ldots,n\}$ corresponds to the leaf $w(i,i)$, which has $O(\log n)$ ancestors. We say that position $i$ is \emph{good} if all its non-leaf ancestors are consistent. The $i$'th output is $w(i,i)$ if $i$ is good, and $1$ otherwise. Since each $w(\ell,r)$ is encoded using $O(\log n)$ bits, this system has locality $O(\log^2 n)$.

Every vector $x_1,\ldots,x_n$ of weight at least $(n+1)/2$ can be generated using this system by taking $w(\ell,r) = x_\ell + \cdots + x_r$. In the other direction, consider any assignment of weights to the tree. If the root is inconsistent, then the output is $1,\ldots,1$, so we can assume that the root is consistent. Prune the tree by removing all children of inconsistent nodes. If $[\ell,r]$ is any node in the pruned tree then either $\ell = r$ and $x_\ell = w(\ell,r)$, or $\ell < r$ and $x_\ell + \cdots + x_r = r-\ell+1 \geq w(\ell,r)$. It follows that $x_1 + \cdots + x_n \geq w(1,n) \geq (n+1)/2$.
\end{proof}

The $\Omega(\log^* n)$ locality lower bound in \cite{BDKMSSTV13} is inspired by the following observation:

\begin{proposition}
\label{prop:input-output-locality}
    Let $f\colon \{0,1\}^m \to \{0,1\}^n$ be such that every output bit is a function of at most $c$ input bits, and every input bit affects at most $d$ output bits. Suppose that $cd \le (n+1)/2$. Then $f(\{0,1\}^m) \neq \Maj^{-1}(1)$.
\end{proposition}
\begin{proof}
    Let $i \in [n]$ be an arbitrary output bit. There are at most $cd$ many output bits in the ``neighborhood'' $N(i)$ of $i$, which is the set of outputs that share an input bit with $i$. If $|N(i)| \leq (n+1)/2$ then we can find an output $y$ of weight $(n+1)/2$ such that $y_{N(i)} = 1^{N(i)}$. Suppose that $y$ is generated by the input $x$. There must be some setting to the inputs of $i$ which sets it to zero (since there is a valid output $z$ with $z_i = 0$). If we modify $x$ using this setting then the new output $z$ agrees with $y$ outside of $N(i)$, and furthermore $z_i = 0$. Since $y_{N(i)} = 1^{N(i)}$, it follows that $|z| < |y|$, which is impossible, since $y$ had the smallest possible weight. 
\end{proof}
One of the steps in our proof (namely, \Cref{lem:S-is-small}) is essentially an adaptation of the proof of \Cref{prop:input-output-locality} for the sources with influential input bits. 

We give a simplified exposition of their proof in \Cref{sec:majority-log-star}. Our own lower bound is contained in the following theorem.

\begin{theorem}
    Let  $f \colon\{0,1\}^m\to\{0,1\}^n $ be a $d$-local function such that $f(\{0,1\}^m) =\Maj_n^{-1}(1)$, where $n$ is odd. Then $d=\Omega(\sqrt{\log n})$.
\end{theorem}
\begin{proof}
    We say that an input bit $i \in [m]$ is \emph{$k$-influential} if at least $k$ output bits of $f$ depend on it. We are going to show that there are $\Omega(n/(dk))$ many $k$-influential bits for every $k$. Then the number of input-output bit pairs where the output depends on the input is at least 
    \( \sum_{k \in [n]} cn/(dk) = \Omega(n \log n / d)\).
    On the other hand, there are at most $nd$ such pairs since $f$ is $d$-local. Therefore $d = \Omega(\sqrt{\log n})$.

    It remains to show the lower bound on the number of $k$-influential bits. This is done by combining
    the following two lemmas.

    \begin{lemma}
        \label{lem:S-is-small}
        Let $I$ be the set of all $k$-influential input bits. Then for every set of output bits $S$ of
        size at most $n/(4 k d)$ there exists an assignment $\rho$ to $I$ such that
        \begin{enumerate}
            \item All bits in $S$ are fixed to $1$ by $\rho$, i.e.\ for every total extension $\rho'$ of
                $\rho$ we have $f(\rho')_S = 1^S$.
            \item $\rho$ fixes to $1$ at most $(n+1)/2$ output bits.
        \end{enumerate}
    \end{lemma}
    
    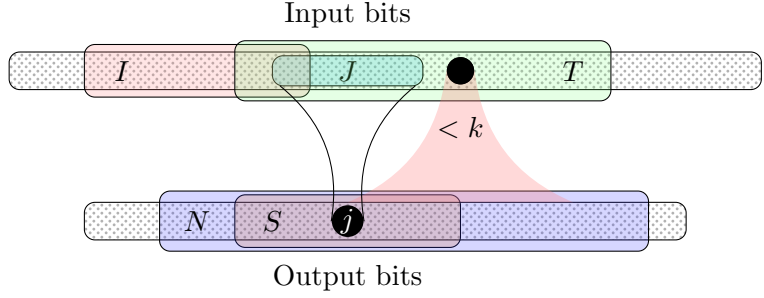
\begin{figure}
        \centering
        \tikzfading[
    name = fade out,
    top color = transparent!0,
    bottom color = transparent!100]

\begin{tikzpicture}
    \node at (4.5, 1) {Input bits};
    \node at (4.5, -2.5) {Output bits};
    \draw[pattern = {crosshatch dots}, pattern color = black!30, rounded corners] (0, 0) rectangle (10, 0.5);
    \draw[pattern = {crosshatch dots}, pattern color = black!30, rounded corners] (1, -2) rectangle (9, -1.5);
    \draw[fill = orange, fill opacity = 0.2, rounded corners] (3, -2.1) rectangle (6, -1.4);
    \draw[fill = blue!80, fill opacity = 0.2, rounded corners] (2, -2.15) rectangle (8.5, -1.35);
    \draw[fill = cyan, fill opacity = 0.2, rounded corners] (3.5, 0.05) rectangle (5.5, 0.45);
    \node at (4.5, 0.25) {$J$};
    \node at (3.5, -1.75) {$S$};
    \node at (2.5, -1.75) {$N$};
    \draw[fill = red!50, fill opacity = 0.2, rounded corners] (1,-0.1) rectangle (4, 0.6);
    \node at (1.5, 0.25) {$I$};
    \draw[fill = green!50, fill opacity = 0.2, rounded corners] (3, -0.15) rectangle (8, 0.65);
    \node at (7.5, 0.25) {$T$};
    \node[circle, fill = black, text = white, radius = 0.25, inner sep = 0.5] at (4.5, -1.75) (a) {$j$};
    \node[circle, fill = black] (b) at (6, 0.25) {};
    \fill[fill = red!50, path fading = fade out, opacity=0.3] (b.west) to[bend left] (4.5, -1.5) -- (7.5, -1.5)
        to[bend left] (b.east);
    \node[circle, fill = black] at (b) {};
    \node at (6, -0.5) {$<k$};
    \draw (a.west) to[bend right] (3.6, 0.05);
    \draw (a.east) to[bend left] (5.4, 0.05);
\end{tikzpicture}
        \caption{$I$ is the set of $k$-influential inputs bits, and $S$ is the given set. The set $T$
            consists of all inputs bits affecting $S$, and the set $N$ consists of all bits influenced by
            $T \setminus I$}
        \label{fig:S-is-small}
    \end{figure}

    \begin{proof}
        Fix a set $S$ of size at most $n/(4kd)$.

        Let $T$ be the set of input bits that affect $S$, so $|T| \le d |S|$. Let $N$ be the set of all
        bits influenced by $T \setminus I$. See \Cref{fig:S-is-small} for a pictorial representation of
        these definitions.

        Since $I$ contains all $k$-influential input bits, $|N| \le kd |S|$. Let $\rho'$ be a total
        assignment such that $f(\rho')_N = 1^N$ and  $|f(\rho')| = (n+1)/2$. This is possible since $|N|
        \le n/4$ by the statement of the lemma. Let $\rho \coloneqq \rho'_I$.

        Since $|f(\rho')| = (n+1)/2$, in particular $\rho$ fixes to $1$ at most $(n+1)/2$ bits. We claim
        that $\rho$ fixes all bits in $S$ to $1$. Suppose for the sake of contradiction that it doesn't
        fix to $1$ the bit $j \in S$. Let $J$ be the set of input bits affecting $j$. Let $\rho''$ be a
        total assignment consistent with $\rho'$ everywhere except $J \setminus I$ such that $f(\rho'')_j
        = 0$. Observe that $f(\rho'')_{[n] \setminus N} = f(\rho')_{[n] \setminus N}$, hence $f(\rho'')
        \le f(\rho')$ coordinate-wise. Since $f(\rho'')_j = 0$ and $f(\rho')_j = 1$, we have $|f(\rho'')|
        < n/2$, which contradicts the fact that the image of $f$ is $\Maj_n^{-1}(1)$.
    \end{proof}

    \begin{lemma}
        \label{lem:S-construction}
        Let $I$ be an arbitrary set of input bits of size at most $n/20$. Then there exists a set $S$ of
        output bits such that $|S| = O(|I|)$ and for every assignment $\rho$ to $I$ that fixes to $1$ at
        most $(n+1)/2$ output bits, there exists a bit in $S$ that is not fixed to $1$ by $\rho$.
    \end{lemma}

    \begin{proof}
        Let $\rho_1, \dots, \rho_K$ be all assignments to $I$ that fix at most $(n+1)/2$ output
        bits. Denote by $U_1, \dots, U_K \subseteq [n]$ the sets of bits that are not fixed by $\rho_1,
        \dots, \rho_K$, respectively, so that $|U_1|, \dots, |U_K| \ge (n-1)/2 \ge n/3$. Then
        \[
            K \frac{n}{3} \le \sum_{i = 1}^K |U_i| = \sum_{j \in [n]} |\{i \in [K] \mid U_i \ni j\}|.
        \]
        Hence there exists $j$ such that $|\{i \in [K] \mid U_i \ni j\}| \ge  K/3$. Let $S_1 \coloneqq
        \{j\}$, and continue this process for the set of bits $[n] \setminus \{j\}$ and the set of
        assignments $\{\rho_i \colon U_i \not\ni j\}$. Suppose the previous iteration yields a set $S_k
        \subseteq [n]$ of size $k$ and a set of indices $T_k \subseteq [K]$. Then let $U'_i \coloneqq U_i
        \setminus S_k$ for $i \in T_k$. Then $|U'_1|, \dots, |U'_K| \ge (n-1)/2 - k \ge n/3$, where the
        last inequality is true if $k < n/6$. As before, there exists $j \not\in S_k$ such that $|\{i \in
        T_k \mid U'_i \ni j\}| \ge |T_k|/3$. We then let $S_{k+1} \coloneqq S_k \cup \{j\}$ and $T_{k+1}
        \coloneqq \{i \in T_k \mid U'_i \not\ni j\}$.

        Clearly $|T_k| \le K \cdot \left(2/3\right)^{k-1}$, hence in $\tau = \lceil\log_{3/2} K\rceil \le
        2 \log_2 K \le 2 |I| \le n/10$ steps we eliminate all assignments from the set, i.e.\ $S_\tau$
        satisfies that for every assignment $\rho$ to $I$ that fixes at most $(n+1)/2$ output bits, there
        exists $j \in S_\tau$ that is not fixed by $\rho$. (The bound $2|I| \leq n/10$ guarantees that
        the condition $k < n/6$ holds).
    \end{proof}

    Let $I$ be the set of all $k$-influential bits. If $|I| \ge n/20$ we get the desired lower bound
    immediately, so assume otherwise. Then let $S$ be the set given by \Cref{lem:S-construction}, $|S| =
    O(|I|)$. By \Cref{lem:S-is-small} we get that $|S| = \Omega(n/(dk))$, so $|I| = \Omega(n/(dk))$ as well.
\end{proof}

\section{Switching Networks for Sampling Linear Slices}
\label{sec:switching}
In this section, we discuss the limitations of switching networks for constructing decision tree samplers. 

\begin{definition}
    A switching network of depth $d$ is a sequence of $d$ matchings $M_1, \dots, M_d$. Each $M_i$ is a set of $n/2$ disjoint pairs of elements from $[n]$. 
    The distribution generated by a switching network over a slice $\binom{[n]}{\ell}$ is defined as follows:
    \begin{itemize}
        \item Initialize the string as $1^\ell 0^{n-\ell}$;
        \item For each $i \in [d]$: for every pair in $(a,b) \in M_i$, we toss a fair coin, and if it comes up heads, we swap the $a$th and the $b$th bits of the current sequence.
    \end{itemize}
\end{definition}

\begin{lemma}[A variation of \cite{Viola12}]
\label{lem:switching-network-to-generator}
    Suppose there exists a switching network of depth $d$ that generates the variable $\bm{X}$ over $\binom{[n]}{\ell}$. Then there exists a decision depth $d$ sampler for $\bm{X}$ such that each input bit affects at most $2^d$ output bits. Moreover, the support of $\bm{X}$ is a subset of $\binom{[n]}{\ell}$.
\end{lemma}
\begin{proof}
    The input bits correspond to the coin tosses in the switching network. Each output bit is computed by tracing back its initial position: first, we query the coin corresponding to the pair in $M_d$ containing the bit, then we query the coin corresponding to the pair in $M_{d-1}$ and so on until we compute the location of the bit in the initial sequence. Then if the location is in $[\ell]$ we output $1$ and otherwise output $0$. It is easy to see that the described sampler has the required properties.
\end{proof}

\begin{lemma}
\label{lem:switching-network-lb}
    Let $\alpha \in (0,1)$ be a constant. Suppose $\bm{X}$ is samplable with a $d$-local sampler such that each input bit affects at most $c$ output bits, the support of $\bm{X}$ is within $\binom{[n]}{\alpha n}$, and $n/(cd)^2 = \omega(2^{cd})$. Then $\Delta(\bm{X}, \bm{U}_{\alpha n}^n) = 1-o(1)$.
\end{lemma}
\begin{proof}
    For each output bit $i \in [n]$ of $\bm{X}$, let $N(i) \subseteq [n]$, the \emph{neighborhood} of $i$, be the set of output bits that share an affecting input bit with $i$. By assumption, $|N(i)| \le cd$. Let us greedily choose a set of output bits with disjoint neighborhoods: $t_1 = 1$, and for $j > 1$, $t_j \in [n]$ is a bit such that $N(t_j) \cap (N(t_1) \cup \dots \cup N(t_{j-1})) = \emptyset$. For each output bit $i$ there are at most $(cd)^2$ output bits $j$ for which $N(i) \cap N(j) \neq \emptyset$, so the greedy process yields $\ell \ge n/(cd)^2$ bits.
    
    Suppose that there is a bit $i \in \{t_1, \dots, t_\ell\}$ such that $\Pr[\bm{X}_i = 0] > 0$ and $\Pr[\bm{X}_{N(i)} = 1^{N(i)}] > 0$. Then we use the approach from \Cref{prop:input-output-locality}: let $x \in \binom{[n]}{\alpha n}$ be a string in the support of $\bm{X}$ such that $x_{N(i)} = 1^{N(i)}$, and let $\rho$ be the input bits that yield $x$. Since $\Pr[\bm{X}_i = 0] > 0$, we can change the bits of $\rho$ affecting the $i$th output bit such that its value switches to $0$. Denote the resulting input by $x'$. Observe then that $x_{[n] \setminus N(i)} = x'_{[n] \setminus N(i)}$, since $N(i)$ is the set of output bits that are affected by the input bits affecting the $i$th bit. Then $|x'| < |x|$, hence $x' \not\in \binom{[n]}{\alpha n}$, so it does not lie in the support of $\bm{X}$, which is a contradiction.
    
    Now let us analyze the case when there are no bits satisfying the condition.
    Let $I \subseteq \{t_1,\dots,t_\ell\}$ consist of those output bits for which $\Pr[\bm{X}_i = 0] = 0$. If $|I| \ge \ell/2$ then
    \begin{multline*}        
     \Delta(\bm{X}, \bm{U}_{\alpha n}^n) \ge \left|\Pr[\bm{X}_I = 1^I] - \Pr[(\bm{U}_{\alpha n}^n)_I = 1^I]\right| = 1 - \binom{n - |I|}{\alpha n} \binom{n}{\alpha n}^{-1}\\
     = 1 - \prod_{j=0}^{\alpha n - 1} \frac{n - |I| - j}{n - j} \ge 1 - \left(1-\frac{|I|}{n-\alpha n  + 1}\right)^{\alpha n} \ge 1 - 2^{-\Omega(\frac{\alpha}{1-\alpha} |I|)} = 1-2^{-\Omega(\ell)}.
     \end{multline*}
     Since $\ell = \Omega(n/(cd)^2) = \omega(1)$, in this case $\Delta(\bm{X}, \bm{U}_{\alpha n}^n) = 1 - o(1)$.

    Now suppose that $|I| \leq \ell/2$, and let $J = \{t_1,\ldots,t_\ell\} \setminus I$. Assume that all bits in $J$ satisfy $\Pr[\bm{X}_{N(i)} \neq 1^{N(i)}] = 1$. Let us compute the probability of this event for $\bm{U}_{\alpha n}^n$:
    \begin{multline*}
        \Pr[(\bm{U}_{\alpha n}^n)_{N(i)} \neq 1^{N(i)}] = 1 - \binom{n - |N(i)|}{\alpha n}\binom{n}{\alpha n}^{-1} = 1 - \prod_{j=0}^{\alpha n - 1} \frac{n - |N(i)| - j}{n - j}\\
        \le 1 - \left(1 - \frac{|N(i)|}{n}\right)^{\alpha n} \le 1 - 2^{-\Omega(\alpha |N(i)|)} = 1 - 2^{-\Omega(cd)}.
    \end{multline*} 
    Therefore
    \[ \Pr[(\bm{U}_{\alpha n}^n)_{N(i)} \neq 1^{N(i)} \text{ for all } i \in J] = \prod_{i \in J} (1-2^{-\Omega(cd)}) \le (1-2^{-\Omega(cd)})^{\ell / 2}.\]
    Since $\ell/2 = n/(2(cd)^2) = \omega(2^{cd})$, we get the desired lower bound on the statistical distance in this case as well. 
\end{proof}

\begin{corollary}
\label{cor:switching-network-lb}
    Let $\alpha \in (0,1)$ be a constant.
    Any switching network that generates a distribution that is $1-\Omega(1)$ close to $\bm{U}_{\alpha n}^n$ has depth $\Omega(\log \log n)$.
\end{corollary}
\begin{proof}
Consider a switching network of depth $d$ that generates a distribution that is $1-\Omega(1)$ close to $\bm{U}_{\alpha n}^n$. \Cref{lem:switching-network-to-generator} translates it to a decision depth $d$ sampler for $\bm{U}_{\alpha n}^n$ such that each input bit affects at most $c = 2^d$ output bits. The sampler is $1-\Omega(1)$ close to $\bm{U}_{\alpha n}^n$, and supported within $\binom{[n]}{\alpha n}$. The result now follows from \Cref{lem:switching-network-lb}.
\end{proof}

\ifanonymous
\else
\paragraph{Acknowledgements.} This project has received funding from the European Union's Horizon 2020
research and innovation programme under grant agreement No~802020-ERC-HARMONIC. 

This work was supported by the Swiss State Secretariat for Education, Research and Innovation (SERI)
under contract number MB22.00026.

We thank Mika G\"o\"os, Aleksandr Smal, and Anastasia Sofronova for fruitful discussions and suggestions.
\fi

\printbibliography

\appendix

\section{Exposition of the \texorpdfstring{$\Omega(\log^* n)$}{log star n} lower bound for Majority}
\label{sec:majority-log-star}
\emph{The following is an exposition of the proof of~\cite[Theorem 5.1]{BDKMSSTV13}.}

Suppose that $n$ is odd, and consider a locality~$c$ proof system for the vectors containing more $1$s than $0$s, that is, having Hamming weight at least $(n+1)/2$. We can assume that $c \geq 2$.\footnote{Alternatively, change $B(0)$ or redefine $d$-influential as having \emph{more than} $d$ output bits depending on it.}

The proof system has \emph{inputs} and \emph{outputs}. The number of outputs is $n$, and each one depends on at most $c$ input bits. An input bit is \emph{$d$-influential} if at least $d$ output bits depend on it. There are at most $cn/d$ many $d$-influential input bits.

Let $1 = B(0) < B(1) < \cdots < B(c+1)$ be a sequence of constants (depending on $c$ but not on $n$), and let $d(\ell) = cn/B(\ell)$, so that $cn = d(0) > d(1) > \cdots > d(c+1) = \Omega(n)$. Thus there are at most $B(\ell)$ many input bits which are $d(\ell)$-influential.

We will construct a sequence of sets $\emptyset = R_0 \subseteq R_1 \subseteq \cdots \subseteq R_c \subseteq [n]$ with the following property: \emph{If $\rho$ is a truth assignment to the $d(\ell)$-influential variables which extends to a complete truth assignment setting all coordinates in $R_\ell$ to zero, then for each coordinate $i \notin R_\ell$, $\rho$ also extends to a complete truth assignment setting coordinate $i$ to zero.}\footnote{In the paper, $\rho$ extends to a complete truth assignment setting both $R_\ell$ and $i$ to zero.}

Given $R_{\ell - 1}$, here is how we construct $R_\ell$. We start with $R \coloneqq R_{\ell - 1}$, and will potentially add more output coordinates to $R$. At any point, suppose that there is a truth assignment $\rho$ to the $d(\ell)$-influential variables which (i) extends to a complete truth assignment setting all coordinates in $R$ to zero, and (ii) for some coordinate $i \notin R$, any complete truth assignment extending $\rho$ sets coordinate $i$ to one. If that happens, then we add $i$ to $R$. Henceforth, $\rho$ will not come up again, since no complete truth assignment extending $\rho$ sets $i$ to one, and $i$ belongs to $R$. Eventually, there is no such ``bad'' truth assignment, and we set $R_\ell \coloneqq R$. 

If $i \in R_\ell$ then there exists some $r \leq \ell$, some $R \subseteq R_r$, and some truth assignment $\rho_i$ to the $d(r)$-influential variables, such that $\rho_i$ extends to a complete truth assignment setting all coordinates in $R$ to zero, and any complete truth assignment extending $\rho_i$ sets $i$ to one.
If $j \in R_\ell$ was added after $i$ then the truth assignment $\rho_j$ extends to a complete truth assignment which sets all coordinates in $R \cup \{i\}$ to zero. In particular, $\rho_j$ doesn't extend $\rho_i$ (as a special case, $\rho_j \neq \rho_i$). It follows that if we extend each $\rho_i$ arbitrarily to a truth assignment to the $d(\ell)$-influential variables, then the resulting assignments will all be different. Consequently,
\[
 |R_\ell| \leq 2^{B(\ell)}.
\]
For large enough $n$, this will be at most $\frac{n-1}{2}$.

We show below that for a proper choice of parameters, there is an output coordinate $i$ which satisfies the following, for all $\ell \in \{0,\ldots,c\}$: \emph{$i \notin R_\ell$, and all inputs to $i$ which are also inputs to $R_\ell$ are $d(\ell+1)$-influential.} We will show that for each $\ell \in \{0,\ldots,c\}$, $i$ has an input which is $d(\ell+1)$-influential but not $d(\ell)$-influential. For different $\ell$ these inputs are different (since an input which is not $d(\ell)$-influential is also not $d(r)$-influential for all $r < \ell$), and so $i$ depends on $c+1$ inputs, which is impossible.

Let $\ell \in \{0,\ldots,c\}$. Let us show that $i$ has an input which is $d(\ell+1)$-influential but not $d(\ell)$-influential. We do this by contradiction: suppose that all $d(\ell+1)$-influential inputs of $i$ are $d(\ell)$-influential. By assumption, $i \notin R_\ell$ and all inputs to $i$ which are also inputs to $R_\ell$ are $d(\ell)$-influential. Let $N(i)$ consist of $i$ together with all other output bits which share some non-$d(\ell)$-influential bit with $i$. All of these shared input bits are in fact non-$d(\ell+1)$-influential, and so
\[
 |N(i)| \leq 1 + cd(\ell+1) \leq 1 + cd(1) = 1 + \frac{c^2n}{B(1)}.
\]
If $B(1) > 2c^2$ then $|N(i)| < 1 + \frac{n}{2}$ and so $|N(i)| \leq (n+1)/2$. Since $|R_c| \leq \frac{n-1}{2}$, the proof system generates some vector $v$ of weight $(n+1)/2$ in which all coordinates of $R_c$ are zero and all coordinates of $N(i)$ are one. Consider an arbitrary complete truth assignment $\alpha$ which generates $v$, and let $\rho$ be its restriction to the $d(\ell)$-influential coordinates. Since $i \notin R_\ell$, by construction, we know that $\rho$ extends to some complete truth assignment $\beta$ which sets $i$ to zero. Now consider the following complete truth assignment:
\[
 \gamma(j) =
 \begin{cases}
 \rho(j) & \text{if $j$ is $d(\ell)$-influential}, \\
 \beta(j) & \text{if $j$ is not $d(\ell)$-influential and influences $N(i)$}, \\
 \alpha(j) & \text{if $j$ is not $d(\ell)$-influential and doesn't influence $N(i)$}.
 \end{cases}
\]
Since $\alpha,\beta$ both extend $\rho$, $\gamma$ agrees with them on the $d(\ell)$-influential variables. Therefore the output generated by $\gamma$ agrees with that generated by $\alpha$ except for the coordinates in $N(i)$, which could change from one to zero. Moreover, the $i$'th output of $\gamma$ agrees with the $i$'th output of $\beta$, namely, it is zero. Therefore the output generated by $\gamma$ has Hamming weight strictly less than $(n+1)/2$, which is impossible.

It remains to show that there exists an output bit $i$ such that $i \notin R_c$, and for all $\ell \in \{0,\ldots,c\}$, all inputs to $i$ which are also inputs to $R_\ell$ are $d(\ell+1)$-influential. We do this by giving an upper bound on the number of bad output bits. An output bit is bad if it either belongs to $R_c$, or for some $\ell \in \{0,\ldots,c\}$, there is a joint input of $i$ and $R_\ell$ which is not $d(\ell+1)$-influential. If $i$ is bad due to some $\ell$, then there must be some non-$d(\ell+1)$-influential input of $R_\ell$ which is an input of $i$. Therefore the number of bad inputs is at most
\[
 |R_c| + \sum_{\ell = 0}^c |R_\ell| \cdot c \cdot d(\ell+1) \leq 2^{B(c)} + n \cdot c^2 \sum_{\ell = 0}^c \frac{2^{B(\ell)}}{B(\ell+1)}.
\]
For a judicious choice of the sequence $B(\ell)$, the coefficient of $n$ will be strictly less than~$1$, and so for large enough~$n$, there are fewer than $n$ bad inputs.

One choice for the sequence $B(\ell)$ is
\[
 B(\ell+1) = 2^{B(\ell)} \cdot 2c^2(c+1).
\]
In particular, $B(1) = 4c^2(c+1) > 2c^2$, which was needed above.
For this choice of $B(0),\ldots,B(c)$, the number of bad inputs is at most
\[
 2^{B(c)} + \frac{n}{2},
\]
which less than $n$ if $2^{B(c)} \leq \frac{n-1}{2}$, a condition which was required at a different step of this proof.

Roughly speaking, $B(\ell+1) \approx 2c^3 \cdot 2^{B(\ell)}$, and so $B(c) \approx 2 \uparrow \uparrow c$. Therefore $2^{B(c)} \leq \frac{n-1}{2}$, and so the argument works, for $c \leq \kappa \log^* n$, for an appropriate constant $\kappa$.
\end{document}